%% file: main.tex
\documentclass[submission,copyright,creativecommons]{eptcs}
\providecommand{\event}{LSFA 2026} 

\usepackage{iftex}

\ifpdf
  \usepackage{underscore}         
  \usepackage[T1]{fontenc}        
\else
  \usepackage{breakurl}           
\fi
\usepackage{packages}
\usepackage{mycommands}

\title{Collusion Relations and their Applications to Balance Theory}
\author{Jean-Baptiste Joinet
\institute{Université Jean Moulin Lyon 3, IRPhiL, Lyon, France}
\email{jean-baptiste.joinet@univ-lyon3.fr}
\and
Carlos Olarte
\institute{Université Sorbonne Paris Nord, LIPN, CNRS, UMR 7030, F-93430, Villetaneuse, France}
\email{olarte@lipn.univ-paris13.fr}
\thanks{
This work was funded by French-Brazilian CAPES-COFECUB Research Programme
    ``Logic and Intelligibility of Computational Processes'' (2024-2027). Olarte
    gratefully acknowledge also the support from the NATO Science for Peace and
    Security Programme through grant number  G6133 (project SymSafe) and the SGR
    project PROMUEVA (BPIN 2021000100160) under the supervision of Minciencias
    Colombia. 
}
}
\def\titlerunning{Collusion Relations and their Applications to Balance Theory}
\def\authorrunning{Jean-Baptiste Joinet \& Carlos Olarte}
\begin{document}
\maketitle

\begin{abstract}
\input{abstract}

\end{abstract}

\section{Introduction}\label{sec:intro}
\input{intro}

\section{Quadrangular Relations}\label{sec:relations}
\input{relations}

\section{The Agonal Interpretation}\label{sec:agonal}
\input{agonal}
\input{protection}

\section{Collusions and Balance Theory}\label{sec:balance}
\input{balance}

\section{A Modal Characterization of Collusions}\label{sec:modal}
\input{logic}

\section{Concluding Remarks}\label{sec:conc}

\input{conclusion}

\nocite{*}
\bibliographystyle{eptcs}
\bibliography{biblio}

\end{document}

%% file: abstract.tex
We study quadrangular properties of binary relations on a set $X$~--i.e., properties defined on
configurations of four elements--~within an agonistic interpretation, where
$xRy$ is interpreted as $x$ ``attacks''~$y$. Such relations
induce a suitable notion of ``protection,'' and we provide necessary and
sufficient conditions for this notion to be consistent. We 
characterize the balance property in signed frames in terms of a specific
quadrangular property, namely collusivity. In this way, we generalize a classical
result in balance theory by offering an alternative method for determining
whether a network is polarized. That is, one can identify well-formed groups of
agents that agree with one another within the same group (a set of allies)
while disagreeing with, or attacking, agents outside the group. Furthermore, we
extend the balance theorem to non-symmetric relations, thereby relaxing a
condition required in standard balance theory. We conclude by giving a modal
characterization of collusive frames, together with corresponding rules in a
labeled sequent calculus, and we show that previous modal characterizations of
balance are derivable within this system.

%% file: intro.tex
Methods for reasoning about group polarization have recently attracted the
attention within the proof theory community, since modal logics have
proven to be useful for analyzing such phenomena. Consider for instance the
Facebook epistemic logic~\cite{DBLP:conf/tark/SeligmanLG13} endowed 
with a symmetric ``friendship'' relation; Tweeting
logic~\cite{DBLP:conf/lori/XiongASZ17} (formalizing announcements in a
network);  the logic for reasoning about social beliefs~\cite{DBLP:journals/synthese/LiuSG14}; the positive-negative modal logic
in ~\cite{DBLP:journals/logcom/PedersenSA21} to reason about friendship and
enmity relations, etc. We may also mention the game-semantic approaches that we
have explored in this context
in~\cite{DBLP:conf/lpar/FreimanOPF24,DBLP:conf/csl/PimentelOLFF25}.

This paper focuses on the problem of \emph{balance} in networks of agents, where
relations may be either positive (e.g., friendship or agreement) or negative
(e.g., enmity or disagreement). The theory of Balance originates in social
psychology~\cite{myers1976group,isenberg1986group} and political
philosophy~\cite{sunstein1999law,sunstein2007group}. It studies how such
relations evolve toward configurations in which triads of agents form 
``stable'' patterns. For instance, if $x$ is a friend of $y$ and $y$ is a
friend of $z$, then the network tends to evolve toward a state in which $x$ and
$z$ are also friends (capturing the intuition that the friend of my friend
should be my friend). We contribute to this line of research in
several ways, including a novel characterization of balance, relaxing a
condition required in ``standard'' Balance theory, and the development of new proof
methods.

The methods proposed in this paper are grounded in the study of quadrangular relations (\Cref{sec:relations}), i.e., properties of a relation $R \subseteq X \times X$ (for a set $X$) expressible by quantifying over four elements of $X$. We focus in particular on \emph{collusive} relations, introduced in~\cite{joinet:hal-02369662,Joinet2021}, defined by the quadrangular condition
$\forall x,y,z,w\in X,\big((xRy \wedge xRz \wedge wRy) \imp wRz\big)$.
Intuitively, whenever $x$ and $w$ share a target $y$, every target of $x$ is also a target of $w$.
Collusions extend equivalence relations, which are precisely reflexive
collusions. Moreover, the standard equi-quotient construction (partitions
induced by equivalence relations) appears as a special case of a more general
quotientation process induced by collusions. In fact, this process is fully
general: collusions characterize
quotients~\cite{joinet:hal-02369662,Joinet2021}. Moreover, the (collusional)
quotients $X/R$ and $X/R^{-1}$ (if $R$ is collusive then its inverse  $R^{-1}$ is collusive too) can be recovered as standard equi-quotients via
two dual \emph{relational indiscernibility} predicates (see~\Cref{eq:rels}).

In~\Cref{sec:agonal}, we endow a binary
relation $R$ with an agonal interpretation, where $xRy$ is understood as $x$ attacks $y$, $x$ disagrees
with $y$, or $x$ is an enemy of $y$. We then provide
necessary and sufficient conditions under which $R$ induces a meaningful notion
of \emph{protection} against attacks. Intuitively, $x$ protects $z$ if $x$ defends $z$
against all its enemies.
Here again a quadrangular property, namely co-confluence,
plays a central role to ensure the \emph{consistency} of the protection relation, in
the sense that no agent can simultaneously protect and attack another.

In~\Cref{sec:balance}, we show a tight connection between collusive relations and 
Balance theory. The classical balance theorem provides three equivalent
characterizations of stable networks (\Cref{sec:classic:balance}). First, a network is stable if it can be
``completed'' by assigning signs to missing edges so that all triads are
stable. Second, this holds iff the network can be partitioned into
groups such that agents within the same group are not negatively  related, while
agents in different groups are not positively related. Third, this is equivalent to
the condition that all cycles in the network satisfy certain 
constraints. In this paper we propose a fourth characterization
(\Cref{sec:new:balance}): the previous conditions are equivalent to check whether 
the positive and negative relations are collusive.

Our approach generalizes standard Balance theory in two complementary
directions. \textbf{First}, classical Balance theory assumes two given
relations (positive and negative) satisfying a fixed set of axioms. In
contrast, we begin with a single relation, the agonal
(attack/disagreement) relation, without imposing any a priori axiom. The
positive relation is then derived from this attack relation, as shown
in~\Cref{sec:balance}. Subsequently, we identify the properties that the attack
relation must satisfy in order 
to have a ``natural'' agonistic interpretation and how this interpretation leads to balanced configurations. 
\textbf{Second}, in Balance theory, both relations are assumed to be symmetric;
additionally, the negative relation is anti-reflexive, while the positive 
relation is reflexive. In our case, anti-reflexivity of the negative relation
is assumed since it has a natural agonistic interpretation. However, 
symmetry is not required (\Cref{sec:b-gen}). In fact, we show that irreflexive
collusive relations (not necessarily symmetric)  induce weakly
balanced networks, while symmetric irreflexive collusive relations yield
strongly balanced frames. These results provide further evidence of the
close connection between collusive relations and stable configurations.

As a final contribution, \Cref{sec:modal} provides a modal
characterization of collusive relations. Following the approach
of~\cite{DBLP:journals/jphil/Negri05,DBLP:journals/apal/MarinMPV22}, the
proposed axiom gives rise to an inference rule in a labeled sequent calculus.
We show that the resulting cut-free system can be used to establish key
properties of balanced networks. In particular, we demonstrate that previous modal 
characterizations of balance in~\cite{DBLP:journals/logcom/PedersenSA21} can
be derived within our system. This provides a new proof-theoretic tool for
reasoning about balance in networks.



%% file: relations.tex
This section explores some properties of relations that will be fundamental for
the results in the forthcoming sections. Along this paper we will assume an
arbitrary set $X$. Given a relation $R\subseteq X\times X$, we write
$xRy$ whenever $(x,y)\in R$. We say that $x\in X$ is a \emph{source} or
\emph{initial} element if its indegree is zero, i.e. $\forall y\in X. \neg
yRx$. Typical properties of binary relations include:

\begin{itemize}
    \item \emph{reflexivity}: $\forall x\in X. xRx$, and 
        \emph{irreflexivity} (or \emph{anti-reflexivity}): $\forall x\in X. \neg xRx$. 
    \item \emph{seriality} (or \emph{totality}):  $\forall x\in X . \exists y\in X. xRy$, and 
             \emph{surjectivity}:  $\forall x\in X . \exists y\in X. yRx$.
    \item \emph{transitivity}: $\forall x,y,z\in X.\left((xRy \wedge yRz) \imp xRz \right)$
        and \emph{anti-trans.}: $\forall x, y, z \in X.(xRy\wedge yRz)\imp \neg xRz$.
\end{itemize}

We will be interested in properties of relations defined on ``quadrilaterals,''
i.e., properties obtained by quantifying on four elements of the set $X$ as
follows. 

\begin{definition}[Quadrangular Relations]\label{def:rel}
    Let $R\subseteq X\times X$. We will say that a property of  the relation
    $R$ is  \emph{quadrangular} iff it is definable by a
    formula, where $R$ is the
    unique free variable, of the following shape:

\[
    \forall x,y,z,w\in X. \left( \left( A_1(x,y) \wedge A_2(x,z) \wedge A_3(y,w)\right) \imp A_4(z,w) \right)
\]

\noindent where $x,y,z,w$ are different variables and the four atomic formulas
of the form $A_i(v,v')$, for $1 \leq i\ \leq 4$, are either $vRv'$ or $v'Rv$
(where $v\neq v'$ and $v,v'\in\{x,y,z,w\}$). 

\end{definition}

According to the above definition, there exist eight different possible
quadrangular properties. Below we list three of them 
that will play an important role in the next sections (see~\Cref{fig:rel}):

\begin{enumerate}
\item $R$ is $\mathbf{Q_1} \defsym \forall x,y,z,w\!\in\!X\,\big((xRy \wedge xRz \wedge yRw) \rightarrow zRw\big)$ ($R$ is called \textbf{confluent})
\item $R$ is $\mathbf{Q_2} \defsym \forall x,y,z,w\!\in\!X\,\big((xRy \wedge zRx \wedge wRy) \rightarrow zRw\big)$ ($R$ is called \textbf{co-confluent} or \textbf{protective})
\item $R$ is $\mathbf{Q_3} \defsym \forall x,y,z,w\!\in\!X\,\big((xRy \wedge xRz \wedge wRy) \rightarrow wRz\big)$ ($R$ is called  \textbf{collusive})
\end{enumerate}

\input{fig-rel}

The property of confluence is well-known.
 Co-confluence represents the
situation where, if there are two possible paths leading to some element $y$
(in the figure one from $z$ and one from $w$), then there is a common
origin for those paths (in the figure $z$). The second name we have used for co-confluence, i.e.,
\emph{protective}, will become clear in the following sections.
As a mnemonic, one 
may think of a collusive relation as a relation where whenever $x$ and $w$
have a common target $y$, then every possible target of $x$ (e.g. $z$ in the
figure) is also a target of $w$. 

Among the eight properties, only
\emph{collusiveness} can be naturally extended to the more general case of a
binary relation $R\subseteq X\times Y$ on two different sets $X$ and $Y$.
This becomes clear when looking at the corresponding diagram 
in~\Cref{fig:rel}, 
where in each node,  either the outdegree is zero or the indegree is zero.
Since the nodes that are sources of arrows are never targets of
arrows,  the vertices can be separated into two disjoint classes: 
the set $X$ of elements that  act as 
sources, and the set  $Y$ of elements  that act as targets.

\begin{definition}[Collusion]\label{def:collusion}
    We say that $R$ is a collusion if $R$ is collusive, total and surjective. 
\end{definition}

%% file: fig-rel.tex
\begin{figure}[t]
\centering
\resizebox{.70\textwidth}{!}{
\begin{tabular}{cccc}

\begin{tikzpicture}[>=stealth]
\node (x) at (0,1) {$x$};
\node (y) at (-1,0) {$y$};
\node (z) at (1,0) {$z$};
\node (w) at (0,-1) {$w$};

\draw[->] (x) -- (y);
\draw[->] (x) -- (z);
\draw[->] (y) -- (w);
\draw[->,dashed] (z) -- (w);

\node at (0,-1.6) {$Q_1$ (Confluence)};
\end{tikzpicture}
&
\begin{tikzpicture}[>=stealth]
\node (x) at (0,1) {$x$};
\node (y) at (-1,0) {$y$};
\node (z) at (1,0) {$z$};
\node (w) at (0,-1) {$w$};

\draw[->] (x) -- (y);
\draw[->] (z) -- (x);
\draw[->] (w) -- (y);
\draw[->,dashed] (z) -- (w);

\node at (0,-1.6) {$Q_2$ (Co-Confluence)};
\end{tikzpicture}
&
\begin{tikzpicture}[>=stealth]
\node (x) at (0,1) {$x$};
\node (y) at (-1,0) {$y$};
\node (z) at (1,0) {$z$};
\node (w) at (0,-1) {$w$};

\draw[->] (x) -- (y);
\draw[->] (x) -- (z);
\draw[->] (w) -- (y);
\draw[->,dashed] (w) -- (z);

\node at (0,-1.6) {$Q_3$(Collusiveness)};
\end{tikzpicture}
&
\begin{tikzpicture}[>=stealth]
\node (x) at (0,1) {$x$};
\node (y) at (-1,0) {$y$};
\node (z) at (1,0) {$z$};
\node (w) at (0,-1) {$w$};

\draw[->] (x) -- (y);
\draw[->] (x) -- (z);
\draw[->] (y) -- (w);
\draw[->,dashed] (w) -- (z);

\node at (0,-1.6) {$Q_4$};
\end{tikzpicture}
\\
\begin{tikzpicture}[>=stealth]
\node (x) at (0,1) {$x$};
\node (y) at (-1,0) {$y$};
\node (z) at (1,0) {$z$};
\node (w) at (0,-1) {$w$};

\draw[->] (x) -- (y);
\draw[->] (x) -- (z);
\draw[->] (w) -- (y);
\draw[->,dashed] (z) -- (w);

\node at (0,-1.6) {$Q_5$};
\end{tikzpicture}
&
\begin{tikzpicture}[>=stealth]
\node (x) at (0,1) {$x$};
\node (y) at (-1,0) {$y$};
\node (z) at (1,0) {$z$};
\node (w) at (0,-1) {$w$};

\draw[->] (x) -- (y);
\draw[->] (z) -- (x);
\draw[->] (y) -- (w);
\draw[->,dashed] (z) -- (w);

\node at (0,-1.6) {$Q_6$};
\end{tikzpicture}
&
\begin{tikzpicture}[>=stealth]
\node (x) at (0,1) {$x$};
\node (y) at (-1,0) {$y$};
\node (z) at (1,0) {$z$};
\node (w) at (0,-1) {$w$};

\draw[->] (x) -- (y);
\draw[->] (z) -- (x);
\draw[->] (y) -- (w);
\draw[->,dashed] (w) -- (z);

\node at (0,-1.6) {$Q_7$};
\end{tikzpicture}
&
\begin{tikzpicture}[>=stealth]
\node (x) at (0,1) {$x$};
\node (y) at (-1,0) {$y$};
\node (z) at (1,0) {$z$};
\node (w) at (0,-1) {$w$};

\draw[->] (x) -- (y);
\draw[->] (z) -- (x);
\draw[->] (w) -- (y);
\draw[->,dashed] (w) -- (z);

\node at (0,-1.6) {$Q_8$};
\end{tikzpicture}
\end{tabular}
}
\caption{
    Quadrangular relations. 
Solid arrows represent the antecedent of the implications in~\Cref{def:rel} and dashed arrows the consequent.\label{fig:rel}}
\end{figure}

%% file: agonal.tex
In this section we interpret a binary relation $R\subseteq X\times X$ as an
``attack'' relation~\S\ref{sec:attack}. We then give necessarily and sufficient
conditions to define a ``protection'' relation~\S\ref{sec:protection} induced
by $R$. Properties relating these relations are
established in~\S\ref{sec:properties} leading to a suitable notion of
a \emph{consistent} protection in~\S\ref{sec:cons-prop}.
As we shall see, collusive and co-confluent relations 
will play a central role to make
our agonal (from the Greek agōn, relating to contest or conflict, as in ``antagonism'')
interpretation sound. 

\subsection{The attack relation}\label{sec:attack}

We will endow the relation $R$ with an \emph{agonal}  interpretation by considering $R$ as an
``attack'' relation in a broad sense. For instance, we may interpret $xRy$ in different ways 
as in: 
argument $x$ is an \emph{attack} against argument $y$;
agent $x$ \emph{criticizes} (or \emph{disagrees} with) agent $y$;
agent $x$ is an \emph{enemy} of $y$;
animal $x$ is a \emph{predator} of animal $y$, etc.
We do not assume a priori any specific property for attack
relations. However, there might be natural desiderata that such 
relations should satisfy to fulfill the intended agonal interpretation. \\

\noindent \textbf{Self-attacks}. 
    If we admit that $xRx$ for some $x$, this means that the argument $x$
    attacks itself, that the agent $x$ criticizes itself, that $x$ is an enemy of
    itself, that the animal $x$ is a self-predator, etc. One may wish to avoid
    the possibility of such
    ``self-agonality'' by prohibiting it, which would correspond to require
    that $R$ is \emph{irreflexive}. However, one could also argue that self-agonality
    exists in different settings. For instance,  a paradoxical argument may
    well be considered as self-refuting (attacking itself); the auto-critic is
    a common phenomena for agents; and even if auto-predation is probably an
    extremely rare phenomena in life, we know that suicide exists.

\noindent \textbf{Mutual attacks}. 
    One may require agonality to be ``mutual'' or ``reciprocal,'' and hence $R$
    to be \emph{symmetric}. For instance, if agent $x$ does not agree with agent $y$, it
    seems perfectly normal to assume also that $y$ does not agree with
     $x$. As another example, assume that $X$ is a set of formulas. One
    could consider that formula $A$ attacks $B$ whenever $B$ is provably
    equivalent to $\neg A$. In a classical world, by contraposition, this would
    also mean that formula $B$ attacks $A$.
    In some cases, however, the symmetry of  $R$ is not desirable. For
    instance, if $X$ is a set of argumentations (derivations which may be
    correct or incorrect as well), then one may consider that an argumentation
    $a$ attacks argumentation $b$ (of type/conclusion $B$) which includes a
    sub-derivation $c$ (of type/conclusion $C$), if $a$ is of type/conclusion
    $D$, where $D$ is provably equivalent to $\neg C$. In this case, 
    it is far to be clear that attacks should be reciprocal. 
    As a further example where symmetry for agonality does not fit
    the reality, consider the predator relation over the set of animals.

\noindent  \textbf{Allied attacks}. 
    As ``union makes power,'' and agents typically act together to improve their
chances of success, it is natural that the attack
relation should be \emph{collusive} in this setting. For instance, if agents $x$ and $w$ 
act together against 
a common
``enemy'' $y$, then it is reasonable that any new enemy of $x$ should also be an enemy of $w$. 
Similarly, for two predators hunting together and pursuing a common prey,
any new prey followed by one of them naturally becomes prey for
the other as well.

In the next section, we analyze the implications of imposing the above
mentioned properties (irreflexivity, symmetry, collusiveness, etc.)  on the attack
relation. We also explore the 
properties inherited by the induced
protection relation defined below. 

%% file: protection.tex
\subsection{The Protection Relation}\label{sec:protection}

Since we interpret the relation $xRy$ as $x$ ``attacks'' $y$, 
a natural question is what the notion of 
``protection'' against such attacks could be.
This subsection defines the notion of protection induced by 
the attack relation and determines the properties that $R$ must satisfy 
for this notion to be ``consistent''.

We first define the notion of protection where, intuitively, 
``$x$ protects $z$'' whenever $x$ defends $z$ from all its enemies, 
i.e., $x$ attacks all $z$'s enemies. More precisely: 

\begin{definition}[Protection $\protects$]\label{def:protection}
 Let $R$ be a binary relation over $X$ (an attack relation). The protection
 relation induced by $R$ is the binary relation $\protects$ over $X$ defined
 by:  
$
        x\protects z \defsym \forall y\in X.(yRz\imp xRy)
$.
\end{definition}

We shall read  $x\protects z$ as  ``$x$ protects $z$,''  ``$x$ is a protector
of $z$,'' ``$z$ is protected by $x$,'' etc.  This means that the interpretation of
$x\protects z$ inherits the agonal flavor of $R$'s interpretation.

As a further intuition, consider the following diagram:
\begin{center}
\begin{tikzpicture}[>=stealth]
\node (y1) at (0,1) {$y_1$};
\node (y2) at (0,0) {$y_2$};
\node (x) at (-1,0) {$x$};
\node (z) at (1,0) {$z$};
\node (y3) at (0,-1) {$y_3$};

\draw[->] (x) -- (y1);
\draw[->] (x) -- (y2);
\draw[->] (x) -- (y3);
\draw[->] (y1) -- (z);
\draw[->] (y2) -- (z);
\draw[->] (y3) -- (z);
\end{tikzpicture}
\end{center}

With the agonal interpretations in mind (refutation, attack, predating, etc.), this  diagram can be read as:
since argument $x$ attacks every argument $y_i$, and those in turn attack argument $z$, then  argument $x$ protects argument $z$;
since agent $x$ criticizes every agent $y_i$, and those   criticize $z$, agent $x$ protects agent $z$;
since animal $x$ preys upon all animal $y_i$, and those are predators of animal $z$, animal $x$ protects  $z$; etc.

However, there is an undesirable phenomena in the definition of the relation $\protects$ 
when we consider initial elements
(see~\Cref{sec:relations}), i.e., elements of $X$ that nobody attacks. More
precisely: 

\begin{observation}[Everybody protects initial elements]\label{obs:all-all}
If $z$ is initial for $R$, then $\forall x\in X .~x \protects z$
\end{observation}
\begin{proof}
Let $z\in X$ be 
initial for $R$ and, to obtain a contradiction, assume that there exists
$x\in X$ s.t. $x\nprotects z$. By definition, there
exists $y\in X$ s.t. $yRz$ and $\neg xRy$. But no such $y$ can
exist, as $z$ is initial.
\end{proof}

The above observation can be read as: ``an element that is attacked by no one
is protected by everyone.'' As this may appear somewhat
counter-intuitive\footnote{A less counter-intuitive reading is that $x$
    protects $z$ against \emph{all} its enemies. If $z$ has no enemies, then
vacuously every $x$ protects $z$ according to \Cref{def:protection}.}, one may
instead require the relation $R$ to be surjective, thereby guaranteeing that
there are no initial elements. For some agonal interpretations, it does not
seem intuitive to require that every element of $X$ be attacked by at least one
other element. This additional restriction, however, makes perfect sense in a
logical context (in the dialectical setting of argumentation), since for every
formula $A$ there exists a formula $\neg A$ that ``attacks'' $A$. Note also
that the theory of collusions provides intrinsic reasons to require
surjectivity (see above for the presentation of the notion of ``collusion'' and
various propositions where surjectivity appears to be a necessary condition).

Instead of demanding $R$ to be surjective, we could refine the notion of protection as follows. 

\begin{definition}[Actual protection induced by $R$]
The actual protection relation induced by $R$ is the binary relation $\aprotects$ over $X$ defined by:  
$ 
    x\aprotects z \defsym (\exists y\in X. yRz) \wedge \forall y\in X.(yRz\imp xRy)
$. 
\end{definition}

When $x\aprotects z$, we say that ``$x$ actually protects $z$,'' or that ``$z$
is actually protected by $x$,'' etc. Clearly, when $x \aprotects z$, $z$ cannot
be initial for $R$. When $x\protects z$ and $z$ is initial (i.e., $x
\not\aprotects z$) we say that $x$ offers a ``virtual'' protection to $z$ (a
protection that, indeed, $z$ does not require). 

\subsection{Properties of Attack and Protection}\label{sec:properties}

In this section we explore some properties of the above defined relations.
First, we show that, in general, attack and protection are completely different
notions. In particular, they are not dual notions, i.e., the 
two notions cannot be considered as opposed/contradictory ones. 

\begin{observation}\label{obs:RandP}
    Depending on the relation $R$ considered, one may well have:

\begin{enumerate*}[label=(\roman*), itemjoin={\quad}]
\item $xRy$ and $x\protects y$
\item $xRy$ and $x\nprotects y$
\item $\neg x{R}y$ and $x\protects y$
\item $\neg x{R}y$ and $x\nprotects y$
\end{enumerate*}
\end{observation}
\begin{proof}
    Below some concrete instances of $X$ and $R$ that support
    the above claims: 

\noindent(i) Let $X=\{1\}$ and $R=\{(1,1)\}$. Hence,  $1R1$ and $1\protects 1$.

\noindent(ii) Let $X=\{1, 2\}$ and $R=\{(1,2)\}$. Hence, $1R2$ but $1\nprotects 2$ (since $\neg 1{R}1$)

\noindent(iii) Let $X=\{1, 2, 3\}$ and $R=\{(1,2), (2,3)\}$. Hence,  $\neg 1{R}3$ but $1\protects 3$.

\noindent(iv) Let $X=\{1, 2, 3\}$ and $R=\{(2,3)\}$. Hence,   $\neg 1{R}2$ and $1\nprotects 3$.
\end{proof}

In general, $x$ does not necessarily protect itself (item (i) below). Additionally, 
the following result shows that 
reflexive protective relations correspond exactly to
symmetric attack relations. 

\begin{observation}[Reflexivity of $\protects$]\label{obs:refl-prot}
Let $R$ be a relation on a set $X$. Then: 
(i) In general, $\protects$ is not reflexive; and 
(ii) $\protects$ is reflexive $\iff$~~~ $R$ is symmetric.
\end{observation}
\begin{proof}
    For the claim (i), 
    let $X=\{1,2\}$ and $R=\{(2,1)\}$. Clearly, $1\nprotects1$.
    For the claim (2), we have: 

    \begin{tabular}{lcl}
        $\protects$ is reflexive &$\iff$ &$\forall y \in X. y\protects y$ (definition of reflexivity)\\
                                 & $\iff$&$\forall y, x\in X . (xRy\imp yRx)$ (definition of $\protects$)\\
                                 & $\iff$&$\forall x, y \in X . (xRy\imp yRx)$ ($\forall$ commutation)\\
                                 & $\iff$ & R is symmetric (definition of symmetry)
    \end{tabular}

\end{proof}

Our definition of protection agrees with the intuition that when $x$ protects
$y$, in general, $y$ does not protect $x$. Additionally, we can show that the
protection relation is not transitive in  general: if we were to assume that
$x\protects z$ follows from  $x\protects y$ and $y\protects z$, we have to
assume that all the enemies of $z$ are also enemies of $y$, which is not always
the case as shown in the following observation. 

\begin{figure}
\centering
\resizebox{.46\textwidth}{!}{
\begin{tabular}{cccc}
\begin{tikzpicture}[>=stealth]
\node (x1) at (0,0) {$1$};
\node (x2) at (1,1) {$2$};
\node (x3) at (2,0) {$3$};
\draw[<->] (x1) -- (x2);
\draw[->] (x2) -- (x3);
\end{tikzpicture}
& \quad  
\begin{tikzpicture}[>=stealth]
\node (x1) at (0,0) {$1$};
\node (x2) at (1,1) {$2$};
\node (x3) at (1,0) {$3$};
\node (x4) at (1,-1) {$4$};
\node (x5) at (0,-1) {$5$};

\draw[->] (x1) -- (x2);
\draw[->] (x2) -- (x3);
\draw[->] (x3) -- (x1);
\draw[->] (x1) -- (x5);
\draw[->] (x4) -- (x5);
\draw[->] (x3) -- (x4);
\end{tikzpicture}
& \quad 
\begin{tikzpicture}[>=stealth]
\node (x1) at (0,1) {$1$};
\node (x2) at (1,1) {$2$};
\node (x3) at (2,1) {$3$};
\node (x4) at (0,0) {$4$};
\node (x5) at (1,0) {$5$};
\node (x6) at (2,0) {$6$};

\draw[->] (x1) -- (x2);
\draw[->] (x2) -- (x3);
\draw[->] (x4) -- (x5);
\draw[->] (x5) -- (x6);
\draw[->] (x4) -- (x2);
\end{tikzpicture}
\\
    $R_1$  & $R_2$ & $R_3$
\end{tabular}
}
\caption{Relations in~\Cref{pr:tr:symm} and~\Cref{obs:col:pro}.\label{fig:rels:obs}}
\end{figure}
\begin{observation}\label{pr:tr:symm}
    In general, $\protects$ is not symmetric, nor transitive. 
\end{observation}
\begin{proof}
    For symmetry, let 
$X=\{1,2, 3\}$ and $R_1$ be as in~\Cref{fig:rels:obs}. Hence, 
  $1\protectsR{R_1} 3$ but $3\nprotectsR{R_1} 1$ (as $2$ attacks $1$, but $3$ does not attack $2$).
  For transitivity, let $X=\{1,2,3,4,5\}$ and $R_2$ be as in~\Cref{fig:rels:obs}. 
  Note that $1\protectsR{R_2} 3$ (which is only attacked by $2$) 
  and $3\protectsR{R_2} 5$ (which is attacked by $1$ and $4$), 
  but $1\nprotectsR{R_2} 5$. Actually none of $5$'s attackers ($1$ and $4$)
  are attacked by $1$ and, even worse,  $1$ does attack $5$. 
\end{proof}

As expected, there might be elements of $X$ that are not protected by anyone
(i.e. protection, in general, is not surjective). 
Moreover,  the surjectivity of $R$ does not
imply the surjectivity of $\protects$. 

\begin{observation}
Let $R$ be a relation on a set $X$. Then: 
(i) In general, $\protects$ is not surjective;

 \noindent (ii) $R$ surjective ~~$\not\hspace{-0,2mm}\imp~~\protects$ surjective.
\end{observation}
\begin{proof}
    Let $X=\{1,2, 3\}$ and $R=\{(1,3)(2,3), (1,1), (2,2)\}$. Note that $R$ is
    surjective. Since no one  protects $3$, $\protects$ is not
    surjective. 
\end{proof}

As hinted above, it is natural to assume that the attack
relation is collusive, i.e.,   agents act together against a common enemy. The
following results show that this property  of $R$ is inherited by the
induced protection relation: if two agents $x$ and $y$ protect agent $z$, then
$x$ will also protect all $y$'s protectees.

\begin{observation}\label{obs:col:pro}
 In general, $\protects$ is not collusive. 

\end{observation}
\begin{proof}
    Let $X=\{1,2, 3, 4, 5, 6\}$ and $R_3$ be as in~\Cref{fig:rels:obs}. 
    In this case,  $1\protectsR{R_3} 3$ and $4\protectsR{R_3} 3$. 
    On the other hand $4\protectsR{R_3} 6$, but $1\nprotectsR{R_3} 6$. 
    So $\protectsR{R_3}$ is not collusive.
\end{proof}

\begin{theorem}\label{th:col-prot}
If $R$ is collusive and surjective, then $\protects$ is collusive. 
\end{theorem}
\begin{proof}
We first note that,  
in general, the collusiveness of $R$ does not imply the collusiveness of
$\protects$. To see that, let $X=\{1, 2, 3, 4, 5\}$ and 
$R=\{(1,2), (2,3)\}$. Clearly, $R$ is collusive  and not
surjective. In particular,  $5$ is initial and from~\Cref{obs:all-all},
 $1\protects 5$ and $4\protects 5$. Since 
  $1\protects 3$ and $4\nprotects 3$, $\protects$ is not collusive.

\noindent Let $R$ be collusive and surjective and let us assume that $x,x', z, z'\!\in\!X$ such that:
({\sf hyp.~1}) $x\protects z$ (i.e.,  $\forall y\!\in\!X.(yRz\;\imp\;xRy)$);
({\sf hyp.~2}) $x'\protects z$ (i.e.,  $\forall y'\!\in\!X.(y'Rz\;\imp\;x'Ry')$); and 
({\sf hyp.~3}) $x\protects z'$ (i.e.,  $\forall w\!\in\!X.(wRz'\;\imp\;xRw)$). 
We want to show that $x'\protects z'$ (i.e.,  $\forall w'\!\in\!X\,(w'Rz'\;\imp\;x'Rw')$).
For that, let us consider a $w'\!\in\!X$ such that $w'Rz'$~~({\sf hyp.~4}) and let us show that $x'Rw'$.
From {\sf hyp.~3}, one gets $w'Rz'\;\imp\;xRw'$ and hence,  by {\sf hyp.~4},  $xRw'$.
As $R$ is surjective, there exists an $y\!\in\!X$ s.t. $yRz$. From {\sf hyp.~1} we know that  $yRz\;\imp\;xRy$  and then,  $xRy$.
Using {\sf hyp.~2} one gets $yRz\;\imp\;x'Ry$ and from the fact that $yRz$, one has $x'Ry$.
Since $R$ is collusive, from $xRy$ , $x'Ry$  and $xRw'$, follows that $x'Rw'$, as needed. 
\end{proof}

\begin{corollary}
If $R$ is collusive  then $\aprotects$ is collusive. 
\end{corollary}

The following theorem justifies the alternative name given to co-confluent
relations namely,  \emph{protective} relations. Consider
configurations $\mathbf{Q_3}$ ($R$ is collusive) and $\mathbf{Q_2}$ ($R$ is co-confluent or
protective) in~\Cref{fig:rel}. In both cases, $x$ and $w$ target $y$. In
$\mathbf{Q_3}$, if $x$ attacks $z$,
then  $w$ must also attack $z$. In $\mathbf{Q_2}$, the arrow
between $z$ and $x$ is oriented in the opposite direction relative to
$\mathbf{Q_3}$. Furthermore, the resulting two-step path from $z$ to $y$ can be
interpreted as $z$ ``protecting'' $y$ against its enemies. We can thus interpret $\mathbf{Q_2}$ as: ``if $z$ protects
$y$ against attacks from $x$, then $z$ must also protect $y$ against attacks
from $w$.''

\begin{theorem}
If $R$ is a co-confluent (a.k.a. protective) collusion, then $\protects$ is a collusion. 
\end{theorem}
\begin{proof}
    Let $R$ be 
    a co-confluent collusion. Then $\protects$ is collusive
    by~\Cref{th:col-prot}. 
    Now let us show that $\protects$ is surjective. For that, let 
 $z\!\in\!X$. Since $R$ is surjective, there exists $y\!\in\!X$ such that $yRz$ and there exists  $x\!\in\!X$, such that $xRy$. We distinguish two cases:
if $y$ is the unique element of $X$ such that $yRz$, we are done, since $x\protects z$; otherwise, 
 let any other $y'\!\in\!X$ s.t. $y'Rz$. Since $R$ is protective, $xRy'$ and we conclude $x\protects z$.
 Finally,  we show that  $\protects$ 
is total. Let $x\!\in\!X$. Since $R$ is total, there exists $y\!\in\!X$ such that $xRy$ and there exists $z\!\in\!X$ s.t. that $yRz$. 
If $y$ is the only element of $X$ such that $yRz$, we are done (since $x\protects z$). Otherwise, 
let $y'$ any other element of $X$ such that $y'Rz$. Since $R$ is protective,  $xRy'$ and therefore, $x\protects z$.
\end{proof}

\subsection{Consistent Protection}\label{sec:cons-prop}

As shown in~\Cref{obs:RandP}, it is possible that $x$ protects $y$ ($x\protects
y$) and, at the same time, $x$ attacks $y$ ($xRy$). This does not seem very
natural and in this section, we establish sufficient and necessary
conditions for the induced protection relation to be ``\emph{consistent},'' in
the sense that a protector of an element is never simultaneously an attacker of
that element. More precisely,

\begin{definition}[Consistent and Complete Protection]
 Let $R\subseteq X\times X$. We say that $\protects$ is consistent iff $\forall x,y\in X.(x\protects y\imp\neg{xRy})$.
Moreover, $\protects$ is consistent and complete iff~$\forall x,y\in X.(x\protects y\iff\neg{xRy})$.
\end{definition}

The following two observations show that the problem of consistency of the induced
protection relation is due to self-attacks. 

\begin{observation}\label{lem:irefl}
Let $x\,\in\,X$. Then, $\neg xRx\imp\forall y\in X.(x\protects y\imp\neg xRy)$.
\end{observation}
\begin{proof}
Let $x,y\in X$ such that $x\protects y$. 
This means $\forall z\in X .(zRy\imp xRz)$. So, in particular,
$xRy\imp xRx$. By contraposition, $\neg xRx\imp \neg
xRy$. Hence if $\neg xRx$, we know that $\neg xRy$.
\end{proof}

\begin{observation}\label{ob:ir-cons}
If $R$ irreflexive  then  $\protects \mbox{is consistent}$. Moreover, 
the converse does not hold in general. 
\end{observation}
\begin{proof}
    The result follows from~\Cref{lem:irefl}.
    To show that the consistency of $\protects$ does not imply the 
    irreflexivity of $R$, 
let $X=\{1,2\}$ and 
 $R=\{(1,2),(2,2)\}$.
 Hence,  $\protects\,=\big{\{}(1,1), (2,1)\big{\}}$, i.e. there is no ``actual
 protection'', only ``virtual protection'' (see~\Cref{obs:all-all} ). 
  Since $\protects\cap\, R=\emptyset$, one has $\forall x,
y\in X.(x\protects y\imp\neg xRy)$, i.e. $\protects$ is
consistent. But $R$ is not irreflexive since $2R2$.
\end{proof}

Now we establish necessary and sufficient conditions
for the induced protection relation to be consistent. 

\begin{theorem}
Let $R$ be co-confluent. Then, 
$\protects\mbox{ is consistent }~~\Leftrightarrow~~R\mbox{ is irreflexive}$. 
\end{theorem}
\begin{proof}
    The ($\Leftarrow$) direction is immediate from~\Cref{ob:ir-cons}. For the 
($\Rightarrow$) direction, assume that $R$ is co-confluent and $\protects$ consistent. 
To obtain a contradiction,  assume that $R$ is not irreflexive. 
and  there exists $x\!\in\!X$ such that $xRx$. 
Since $\protects$ is consistent, we must then have $\neg \,x\!\protects\! x$.
Hence, by $\protects$'s definition, there exists a $y\!\in\!X$ such that
$yRx$ and $\neg xRy$.
Since $R$ is co-confluent,
we have $\big{(}(yRz\, \wedge\, y'Rz\,\wedge\,xRy) \imp xRy'\big{)}[x/y\,, x/z\,, y/y'\,, y/y]$, i.e., 
$\big{(}(xRx\, \wedge\, yRx\,\wedge\,xRx) \imp xRy\big{)}$. 
Given that $xRx$ and $yRx$, we must have $xRy$, which is a contradiction. 
\end{proof}

The following observation investigates whether some other properties
of $R$ are inherited or not in the induced protection 
relation. 

\begin{observation}
The following holds: 
(i) $R$ anti-transitive $\imp$ $R$  irreflexive; 
(ii) $R$ anti-transitive \;$\imp\;\, \protects \mbox{consistent}$; 
(iii) $R$ irreflexive $\not{\hspace{-1,3mm}\imp}$ $R$ anti-transitive;  and 
(iv) $\protects \mbox{consistent}~\not{\hspace{-1,3mm}\imp}~R$ anti-transitive.
\end{observation}
\begin{proof}
\noindent\textbf{(i)}
We have $\forall x\in X.\big{(}(xRx\wedge xRx)\imp \neg xRx)\big{)}$ and therefore,  $\forall x\in X. \neg xRx$.
Claim \textbf{(ii)} follows 
directly from (i) and~\Cref{ob:ir-cons}. 
For claim \textbf{(iii)}, note that  
$R=\{(1, 2), (2,3), (1,3)$ is irreflexive but not anti-transitive.
Finally, for \textbf{(iv)}, let 
$R=\{(1, 2), (2,3), (1,3), (4,3)\}$, which  is not anti-transitive but  $\protects$ is consistent. 
Indeed, it happens that the only protected elements are the initial ones, i.e., the non attacked ones:
$
\protects=\emptyset\,\cup\,
\big{\{}(1,1), (2,1), (3,1), (4,1)\big{\}}\,\cup\,\big{\{}(1,4), (2,4), (3,4), (4,4)\big{\}}
$. 
Since those elements are attacked by no one, no situation involving them can contradicts consistency.
\end{proof}

We conclude by noticing that, from the desiderata in~\Cref{sec:attack}, we do
not necessarily impose symmetry in $R$ to obtain a consistent protection
relation. Moreover, from \Cref{obs:refl-prot}, we know that if $R$ is symmetric, then
$\protects$ is necessarily reflexive. 

%% file: balance.tex
Balance theory originates from theories in social psychology~\cite{heider}.
This theory asserts  that connections between friends and enemies tend to be in
a balance state, where  ``unstable'' configurations, as triangles of agents of
two friends and one enemy, are likely to disappear. In the context of
\emph{social network logics} (see e.g.,
\cite{DBLP:journals/logcom/PedersenSA21,DBLP:conf/lpar/FreimanOPF24}), (part
of) the theory of balance has been axiomatized within the framework of
\emph{signed frames} where there is a symmetric \emph{positive} relation
(between friends) and a symmetric \emph{negative} relation (between enemies). A
fundamental result in balance theory~\cite{balance} is how to link the above
mentioned \emph{local} property between three-agents configurations to the
formation of antagonic groups of agents in the network. In this section we show
how the agonal interpretation in~\Cref{sec:agonal} allows us to give an
alternative characterization of the balance theorem, by relating  collusive
relations  with balanced configurations. Moreover, we generalize the conditions
in the balance theorem to the non-symmetric case. 

\subsection{Signed Frames and the Balance Theorem}\label{sec:classic:balance}

A signed frame (or network) is an undirected graph where nodes are agents (in our case,
elements of an arbitrary set $X$) with two binary relations: $\rplus$ linking
friends, and $\rminus$ linking enemies. In a more general setting, $x\rplus y$
not only means that ``$x$ is a friend of $y$,'' but also that ``$x$ agrees with
$y$'' in a particular topic of discussion. Similarly, $x\rminus y$ stands for
``$x$ is an enemy of $y$,'' ``$x$ does not agree with $y$,'' etc. 

\begin{definition}[Signed Frame]\label{def:signed}
A \emph{signed} frame $\mathcal{F}$ is a tuple $\langle X,\rplus,\rminus\rangle$ where $X$
is a set (of agents),  
$\rplus,\rminus\subseteq X\times X$, and:
(1)  $\rplus$ is reflexive and symmetric; and (2) 
 $\rplus$ and $\rminus$ are non-overlapping, i.e., $\forall x,y\in X: \neg(x\rplus y) \vee \neg(x \rminus y)$. 
 We say that $\mathcal{F}$ is a \textbf{symmetric signed frame} (for short \ssf), if additionally (3) 
	$\rminus$ is symmetric. 
Finally, a (symmetric) signed frame 
is said to be \textbf{collectively connected} (for short \cc) if additionally 
	 $\forall x,y \in X.  (x\rplus y\vee x \rminus y)$.
\end{definition}

In ``standard'' balance theory, only the case where $\rminus$ is symmetric is
considered. Hence, the balance theorems in this and the following section focus on symmetric signed
frames (\ssf). 
In \S\ref{sec:b-gen} we generalize this setting to signed frames where
$\rminus$ is not necessarily symmetric.

The network of agents in~\cite{DBLP:journals/logcom/PedersenSA21} is modeled as an
undirected graph and then, the positive and negative relations need to be
symmetric. It is assumed that any agent $x$ \emph{agrees} with itself and
hence, $\rplus$ is reflexive. Moreover, agents $x$ and $y$ cannot be, at the
same time, friends and enemies and the two relations are non-overlapping.
Hence, by symmetry of $\rplus$ and non-overlapping, $\rminus$ is necessarily
irreflexive. In some cases, it is assumed that the network is
collectively connected (\cc), i.e., all the agents in the network are positively or
negatively related. This extra condition is not too restrictive from the point
of view of balance since, as shown in~\Cref{thm:balance}, a balanced \ssf~ can
be always extended to a balanced \cc~\ssf

\paragraph{Triads and stable configurations.}

Consider three agents related with positive or negative relations. There are 
four possible  configuration as depicted
in~\Cref{fig:balance}. Some of these configurations are called ``stable'',
since they induce stability properties. More precisely, 
stable configurations are either triads where all the relations are
positive (and the three agents are friends, as in (a)), or configurations with two
negative relations and one positive relation (case  (b) in the figure).
Configuration (a)  corresponds to the situation where ``the friend of my
friend is my friend,'' and configuration (b) correspond to the situation
where  ``the enemy of my enemy is my friend.''
Configuration (c) and (d) are unstable since they entail certain ``tension''
between the agents. For instance, in (c), there is a tension for two of the
agents to become friends, e.g. $x$ and $y$,  ``against'' the common enemy $z$.
Similarly, in (d), there is a tension for $x$ and $y$ to become friends. 


\begin{definition}[Local Balance]\label{def:local-balance}
The frame $\mathcal{F}=\langle X,\rplus,\rminus\rangle$ 
has the \emph{local balance} property iff for all $x,y,z$:
    (1) $ ( (x\rplus y \wedge y\rplus z) \vee  (x\rminus y \wedge y\rminus z)) \imp x\rplus z$~~~~and~~~~
    (2) $( (x\rplus y \wedge y\rminus z) \vee (x\rminus y \wedge y\rplus z)) \imp x\rminus z$.
\end{definition}

A signed frame thereby has the local balance property iff the graph do not include the triads (c) and (d). 
However, balance can be too restrictive in some situations,
and the notion of weak balance~\cite{davis-weak}
accepts configuration (c).

\begin{definition}[Local Weak Balance]\label{def:local-weak-balance}
A frame  $\mathcal{F}=\langle X,\rplus,\rminus-\rangle$ 
has the \emph{local weak balance} property iff for all $x,y,z \in X$: 
(1)  $ (x\rplus y \wedge y\rplus z) \imp x\rplus z$~~~~and~~~~ (2) 
	  $( (x\rplus y \wedge y\rminus z) \vee (x\rminus y \wedge y\rplus z)) \imp x\rminus z$.
\end{definition}

\begin{figure}
    \begin{center}
\begin{tabular}{cccccc}
\pnlfig{+}{+}{+} & 
\pnlfig{-}{-}{+} & 
\pnlfig{-}{-}{-} & 
\pnlfig{+}{+}{-} \\
(a) & (b)& (c)& (d)
\end{tabular}
    \end{center}
\caption{
Configurations (a) and (b) are balanced. 
Configurations (a), (b) and (c) are weak balanced. 
 \label{fig:balance}}
\end{figure}

A signed frame thereby has the local weak balance property iff the graph do not include the triad  (d). 
Below, to avoid confusions, we sometimes use ``strong balance'' to the denote the 
property in~\Cref{def:local-balance}.

\begin{theorem}[Balance~\cite{balance}]\label{thm:balance}
Let $\mathcal{F}=\langle X,\rplus,\rminus \rangle$ be a \ssf. The three
following properties are equivalent:
\begin{enumerate}
	\item There exists a \cc  \ssf~ $\mathcal{F'}=\langle X, \R_1^+,\R_1^-\rangle$
	such that $\R^+ \subseteq \R_1^+$ and $\R^- \subseteq \R_1^-$
	that has the local balance property (i.e., all the triangles in $\cF'$ satisfy the local balance property). 
	\item There is a bi-partition $S_1$, $S_2$ of $X$ s.t. $\forall x,y\in X$:
         If $x\rplus y$, then $x,y\in S_i$ for some $i\in\{1,2\}$; and 
		 If $x\rminus y$, then $x\in S_i$ and $y \in S_j$ with $i\neq j$. 
	\item All simple cycles in $\mathcal{F}$ have an even number of negative edges.
\end{enumerate}
\end{theorem}

The Balance Theorem above, which assumes $\rminus$ to be symmetric,
establishes a correspondence between a local balance condition and a \emph{global} notion
of balance. In particular, it shows that a balanced frame can be partitioned into two
antagonistic groups such that agents within each group are either unrelated or
friends (never enemies), while agents across groups are either unrelated or
enemies (never friends).
Moreover, this holds iff the frame can be ``completed'' into a \cc~network in
which all triangles satisfy the local balance property.
Consequently, in any balanced \cc~network,
the two resulting groups satisfy that any 
agents within the same group are
positively related, whereas any agents belonging to different groups are negatively
related. \footnote{
    This corresponds to a form of \emph{polarization}, where agents can be
    partitioned into opposing groups.
}
Finally, the above two  conditions are equivalent to the requirement that every simple cycle
of the graph (i.e., a path whose first and last nodes coincide and whose
intermediate nodes are distinct) contains an even number of negative relations.
Observe that  only configurations (a) and (b) in~\Cref{fig:balance} satisfy this
property.

\begin{theorem}[Weak Balance~\cite{davis-weak}]\label{thm:weak}
Let $\mathcal{F}=\langle X,\rplus,\rminus \rangle$ be a \ssf. The three
following properties are equivalent:
\begin{enumerate}
	\item There exists a \cc \ssf~ $\mathcal{F'}=\langle X,\R_1^+,\R_1^-\rangle$
	such that $\rplus \subseteq \R_1^+$ and $\rminus \subseteq \R_1^-$
	with the weak local balance property (i.e., all the triangles in $\cF'$ have the weak local property). 
\item There exists a partition $\{S_i\}_{i\in I}$ of $X$  such that $\forall x,y\in X$: 
		If $x\rplus y$, then $x,y\in S_j$ for some $j\in I$; and 
		If $x\rminus y$, then $x\in S_j$ and $y \in S_k$ for some $j\neq k$. 
	\item There are no cycles in $\mathcal{F}$ with exactly one negative edge. 
\end{enumerate}
\end{theorem}


 In this case, a weakly balanced network can be partitioned
into $n$ groups rather than just two. Within these groups, agents are unrelated or friends,
while negative relations are present between members of different groups. This
is a direct consequence of accepting configuration (c) in~\Cref{fig:balance},
where agents $x$, $y$, and $z$ can each belong to a distinct group, a situation
that is not possible in the (strong) balance case. 

\subsection{A new Characterization for Balance}\label{sec:new:balance}

This section presents a novel and fourth characterization of the strong and weak balance
theorems in terms of collusions.

\begin{theorem}[Balance]\label{th:new-balance-strong}
Let $\mathcal{F}=\langle X,\R^+,\R^-\rangle$ be a \cc\ssf. Then, 
$\mathcal{F}$ is balanced iff both $R^+$ and $R^-$ are collusive. 
\end{theorem}
\begin{proof}
($\Rightarrow$ side).  Assume that $\mathcal{F}$ is balanced. By
Theorem \ref{thm:balance}, and the fact that $\cF$ is \cc, there exists a bi-partition $S_1, S_2$
of $X$ s.t. two elements in $S_i$ ($i\in \{1,2\}$)  are $\R^+$-connected and 
elements in  $S_i$ and $S_j$, $i\neq j$,  are $R^-$-connected. 
To show that $\rminus$ is collusive assume that $x\rminus z$, $y\rminus z$ and $x\rminus w$.
Due to balance, it must be the case that $x\rplus y$.
Since $x\rminus w$ and  $x\rplus y$, by balance, 
$y\rminus w$ as needed. 
To show that $\rplus $ is collusive assume that $x\rplus z$, $y\rplus z$ and $x\rplus w$.
By balance, $x\rplus y$ and also $y\rplus w$ as needed. 

\noindent($\Leftarrow$ side). Assume that both $\rplus $ and $\rminus $ are collusive and, to find a 
contradiction, 
that $\mathcal{F}$ is not balanced. Then, there exists a triangle that does
not satisfy the local balance property and there are two cases:

\noindent \textbf{(i)} 
     Configuration (c) in Figure \ref{fig:balance}. Since $\rminus $
        is symmetric, and collusive, it must be the case that $x\rminus x$, which is not possible
        since $\rminus$ is irreflexive; and 

\noindent \textbf{(ii)} 
   Configuration (d). Since $\rplus $ is reflexive, 
        we have $x\rplus z$, $y\rplus z$ and $x\rplus x$. Since $\rplus $ is collusive, 
        it must be the case that $x\rplus y$, which is not possible since the relations
        are  non-overlapping. 
\end{proof}

Note that if  $\rminus=\emptyset$, trivially $\mathcal{F}$ is balanced, since
all the relations will be as in configuration (a) in~\Cref{fig:balance}. The next
corollary shows that the collusiveness of both $\rplus$ and $\rminus$, when $\rminus$ is not
empty, are necessary and sufficient conditions for a \cc\ssf~to have the
balance property.

\begin{corollary}
Let $\mathcal{F}=\langle X,\rplus ,\rminus \rangle$ be a \cc~\ssf~
and $\rminus \neq \emptyset$. 
$\mathcal{F}$ is balanced iff $R^+$ and $R^-$ are both collusions.
\end{corollary}
\begin{proof}
    By~\Cref{th:new-balance-strong} we know that $R^+$ and $R^-$ are both collusive. 
    By reflexivity of $\rplus $, this relation is also a collusion. 
    If $\rminus \neq \emptyset$, the existence of the bi-partition is guaranteed 
    (in virtue of~\Cref{thm:balance}), 
        where all the elements of one group are $\rminus $-connected to 
    all the elements in the other one. Therefore, 
    $\rminus $ is surjective and total and,
    by~\Cref{th:new-balance-strong}, it is a collusion. 
\end{proof}

Now we prove a similar result characterizing the weak-balance property. In this
case, our theorem  shows that the collusiveness of $\rplus$ is a sufficient and
necessary condition for the frame to have this property. 

\begin{theorem}[Weak Balance]\label{th:new-weak-balance}
Let $\mathcal{F}=\langle X,\rplus ,\rminus \rangle$ be a \cc~\ssf~
$\mathcal{F}$ is weak balanced iff $R^+$ is collusive. 
\end{theorem}
\begin{proof}
($\Rightarrow$ side).  Assume that $\mathcal{F}$ is weakly balanced. By
Theorem \ref{thm:weak}, there are  partitions $S_1,\cdots, S_n$
with the expected properties. 
To show that $\rplus $ is collusive, assume that $x\rplus z$, $y\rplus z$ and $x\rplus w$.
By (weak) balance, $x\rplus y$ and also $y\rplus w$ as needed. 
The ($\Leftarrow$) side  is 
identical to the ($\Leftarrow$ side) in the proof of~\Cref{th:new-balance-strong}.

\end{proof}

\noindent Since  $\rplus$ is total and surjective, the following result
follows immediately from the previous theorem.

\begin{corollary}
Let $\mathcal{F}=\langle X,\rplus ,\rminus \rangle$ be a \cc~\ssf~
$\mathcal{F}$ is weak balanced iff $\rplus$ is a collusion.
\end{corollary}

\subsection{Generalization of Balance}\label{sec:b-gen}
This section shows how to generalize the (weak) balance theorem 
to the  case where the negative relation is not 
necessarily symmetric. 
Additionally, we show how irreflexive and symmetric
collusions induce balanced configurations. 

\begin{definition}[Classes and equi-targeting relations)]\label{eq:rels}
    We use $[x]_R$ to denote the set $\{y\in X \mid xRy \}$. 
    Moreover, we define the following relations induced by $R$:
    \[
        x \friendsL_R y \defsym \forall z\in X.(zR x \iff z Ry)
        \qquad
        x \friendsR_R y \defsym \forall z\in X.(xR z \iff y Rz)
    \]
\end{definition}

In $x\friendsL y$ we say that $x$ and $y$ are equi-targeted in $R$
and we can imagine $x$ and $y$ to be part of a \emph{coalition} for
defending themselves from their
common enemies. 
In $x\friendsR y$ we say that $x$ and $y$ are equi-targeters in $R$, 
and these agents are in a \emph{coalition}  to attack all their common targets. 

\begin{observation}\label{fr:equiv}
     $\friendsL_R$ and $\friendsR_R$ are  equivalence relations. Moreover, 
    if $R$ is symmetric then  $\friendsL_R = \friendsR_R$. Hence, in the case
    of a symmetric relation, we denote with  $\friends_R$ the induced relation
    in~\Cref{eq:rels}. 
\end{observation}
\begin{proof}
    Trivial from the definition. 
\end{proof}

The next theorem shows that irreflexive collusions induce partitions
on the set $X$  where elements of the same partition do not attack
to each other.
\begin{theorem}\label{th:irr-col}
    If $R$ is an irreflexive collusion, then:
    (1) the set $\{[x]_R~\mid~ x\in X\}$ is a partition of $X$; and 
    (2) for all $x,y,y'\in X$, if $y,y' \in  [x]_R$ then  $\neg yR y'$.  
\end{theorem}
\begin{proof}
    By~\cite[Theorem 2]{joinet:hal-02369662}  ($R$ is a collusion iff $\{[x]_R\}$ is a partition)
    we conclude (1). 
    By~\cite[Lemma 13]{joinet:hal-02369662}  (for any irreflexive collusion $R$, targeting class $[x]_R$ 
        and elements $y,y'\in [x]_R$, $\neg y R y'$) we conclude (2). 
\end{proof}

\begin{example}
    Let \( X = \{1, 2, 3\} \) and \( R = \{(1,2), (2,3), (3,1)\} \).
Since \( R \) is an irreflexive collusion, it follows from~\Cref{th:irr-col}
that \( X \) can be partitioned
as $\{[1]_R, [2]_R, [3]_R\}$ (where elements of the same class do not attack to each other). 
Since $R$ is not symmetric, \Cref{thm:weak} cannot be applied here to guarantee the existence of such a partition.
\end{example}


%

The next theorem shows that irreflexive collusions induce a 
weakly balanced signed frame (where $\rminus$ is not necessarily symmetric). 

\begin{theorem}[Generalized weak balance]\label{th:new-balance}
    Let $R$ be an irreflexive collusion on $X$. Then the structures $\cF= \langle X, \friendsL_R, R \rangle$
    and $\cF= \langle X, \friendsR_R, R \rangle$
    are weakly balanced signed frames. 
\end{theorem}
\begin{proof}
    We are defining $\rplus$ as $\friendsL_R$ and $\rminus$ as $R$. 
    Due to~\Cref{fr:equiv}, 
    $\rplus$ is reflexive and  symmetric.
    \textbf{Non-overlapping}: 
    To reach a contradiction, assume that 
     $x\friendsL_R y$ and $xRy$. Then, $xRx$ and 
     $R$ cannot be irreflexive. 
    \textbf{Weak balance}: 
    By~\Cref{th:irr-col}, 
    $\{[x]_R\}_{x\in X}$ is a partition of $X$.
    Assume that $x\rminus y$, i.e., $xRy$. In this case, $y\in [x]_R$ and, 
    due to irreflexivity of $R$, $x\not\in [x]_R$. 
    Hence, $x$ and $y$ are necessarily in two different classes. 
    Now assume that $x\rplus y$, i.e., $x\friendsL_R y$.
    Since $R$ is surjective, 
    there exists $z\in X$ s.t. $zR x$. Since
    $x\friendsL_R y$, we  also have $zR y$
    and then, $x,y\in [z]_R$. 
    The case $\friendsR_R$  reduces to the former case. Indeed, By~\cite[Lemma 7]{joinet:hal-02369662},
    $\friendsR_R = \friendsL_{R^{-1}}$ and by~\cite[Lemma 6]{joinet:hal-02369662},
    collusions are closed by converse.
\end{proof}

\begin{corollary}
    Let $R$ be an irreflexive collusion. Then the structures 
    $\cF= \langle X, \friendsL_R, R\cup R^{-1} \rangle$ and 
    $\cF= \langle X, \friendsR_R, R\cup R^{-1} \rangle$
    are weakly balanced \ssf 
\end{corollary}
\begin{proof}
    By~\Cref{th:new-balance}, $\cF= \langle X, \friendsL_R, R \rangle$
    satisfies the conditions to be weakly balanced.
    What we do now is to ``complete'' the frame with the missing 
    links (i.e., $R^{-1}$) to make $\rminus$ symmetric.
\end{proof}

\paragraph{The (strong) balance case.}
We can show that configuration (c) in~\Cref{fig:balance} is impossible 
if $R$, modeling the negative relation,  is a symmetric irreflexive collusion. The next theorem 
shows something stronger: the absence of cycles 
in $R$ of odd length. 

\begin{theorem}\label{th:odd}
Let $R$ be a symmetric, irreflexive collusion. Then $R$ cannot form cycles of odd length. 
\end{theorem}
\begin{proof}
  We show by induction on $n$ that $\forall n\!\in\!\mathbb{N}$, there is no cycle having length $2n+1$.
The case  n=0 is trivial  since there is no cycle of length 1 since $R$ is irreflexive. 
%
Case $n=m+1$.
The induction hypothesis is that there is no cycle with length $2m+1$.
To reach a contradiction, let us assume that there exists an odd cycle $c$ of
length $2n+1$. Let $x$ be an element in $c$. We consider an
arbitrary orientation of $c$ and then the path $p$ from $x$ to $x$ ``through''
$c$. Note that the length of $p$ is  $2(m+1)+1$
and 
there exists an oriented subpath $p'$ of $p$, with length
$2m$, where  $p=(x,y)(y,z)p'(w,x)$.  Since $R$
is symmetric, 
we  have  $(x,y), (z,y), (x,w)\!\in\!R$. Moreover, since $R$ is 
collusive, we also have $(z,w)\!\in\!R$.
Therefore, 
the path $(z,w)p'$ is a cycle of  length $2m+1$, which contradicts our induction hypothesis. 
\end{proof}

It is easy to see that configuration (d) in~\Cref{fig:balance}  is not possible if
the positive relation is an equivalence relation. In the absence of
configurations (c) and (d), we can form a balanced frame as follows. 

\begin{theorem}\label{th:str-balance}
 If $R$ is a symmetric irreflexive collusion, then $\langle X,  \friends_R, R\rangle $ is a (strong) balanced frame. 
\end{theorem}
\begin{proof} 
    $\friends_R$ is reflexive and symmetric (\Cref{fr:equiv}).
    $R$ is symmetric. If $x\friends_R y$ and $xRy$, we have $xRx$, which is not possible
    since $R$ is irreflexive. This shows that $\friends_R$ and $R$ are
    non-overlapping. Hence,   $(X, \friends_R, R)$ is a signed frame in the
    usual sense. By~\Cref{th:odd}, 
    there are no (c) triads.
     Moreover, 
    by non-overlapping and transitivity of $\friends_R$, 
    triads of type (d) are not possible. Hence, every triad satisfies the (local) balance
    condition. 
\end{proof}

%% file: logic.tex
In this section, we give a modal characterization of collusive relations. 
For that, we consider a multi-modal logic with two modal operators, namely $\RBOX$
and $\LBOX$. As usual, these connectives, together with negation, give rise to
the dual operators $\RDIAM$ and $\LDIAM$. We then show how the quadrangular
relations studied here induce inference rules for a cut-free labeled sequent
calculus, which we use to establish some properties of balanced frames.

As usual~\cite{DBLP:books/cu/BlackburnRV01}, 
modal formulas are interpreted in a Kripke structure
$\cM = \langle X,R,V \rangle$, where $X$ is a non-empty set of words, $R\subseteq X\times X$, and 
$V$ is a labeling function interpreting
the atomic symbols at each world in $X$. 
The semantics is the standard one and we
only present the case for the modal connectives:
\[
    \begin{array}{lll}
        \cM,x \models \RBOX \phi & \mbox{ iff } & \mbox{ for all } x' \mbox{ s.t. } xRx'\mbox{, } \cM,x' \models \phi\\
        \cM,x \models \LBOX \phi & \mbox{ iff } & \mbox{ for all } x' \mbox{ s.t. } x'Rx\mbox{, } \cM,x' \models \phi
    \end{array}
\]

The forcing relation is extended to frames as usual: 
Given a frame $\cF=\langle X,R\rangle$, 
$\cF \models \phi$ iff 
for all valuation $V$ and world $x\in X$, $\langle X,R,V\rangle, x\models \phi$. 

It is clear from the above semantic rules 
that $\RBOX$ is interpreted on the
relation $R$ (by considering all the $R$-successors) while $\LBOX$ is
interpreted on $R^{-1}$ (by considering all the $R$-predecessors). 
These modalities allow us to 
give a modal characterization of collusive frames
via the following axiom scheme: 
\[
    \mbox{Axiom }\axC:\RDIAM \LDIAM \RBOX A \imp \RBOX A
\]

\begin{theorem}[Modal characterization]\label{th:modal:char}
    Let
     $\cF=\langle X, R\rangle $ be a frame. 
    $R$ is collusive iff $\cF \models \axC$.
\end{theorem}
\begin{proof}
($\Rightarrow$ side). Let $R$ collusive, $x\in X$, $V$ any valuation, 
and $\cM =\langle X,R, V \rangle$. Assume also that
$\cM,x \models \RDIAM \LDIAM \RBOX A$
and let us show that 
$\cM,x \models \RBOX A$. 
Since $\cM,x \models \RDIAM \LDIAM \RBOX A$, 
there exist  
$y,z\in X$ s.t. $xR y$, $zRy$ and $z,\cM \models \RBOX A$. 
We conclude by noticing that for all $w\in X$, $zRw$ iff $xRw$ (since $R$ is collusive)
iff 
$\cM,w \models A$ (since $\cM,z \models \RBOX A$)
iff $ \cM,x \models \RBOX A$ (by definition). 

\noindent($\Leftarrow$ side). 
Let $\phi'$ be the axiom scheme 
$\RDIAM A \imp \RBOX \LBOX \RDIAM A$ (the contrapositive of $\axC$). 
To find a contradiction, let $x,y,z,w\in X$ and assume that  $xRy$, $zRy$, $zRw$ 
and $\neg xR w$ (i.e., $R$ is not collusive). Let $q$ be an atomic proposition and 
assume a valuation $V$ that makes $q$ true only at $w$. 
In the model $\cM=\langle X,R,V\rangle$, we have 
$\cM,z \models \RDIAM q$. 
However, 
$\cM,z\not\models \RBOX\LBOX\RDIAM q$,
since $z R y$, $xR y$ and 
$\cM,x \not\models \RDIAM q$.
\end{proof}

\subsection{Proof Theory for Collusions}
Label sequent calculi internalize the Kripke semantics of modal logics into 
the inference rules. Hence, 
 the forcing relation $\cM,x \models A$  is ``replaced'' with an internal relation written as $\lf{x}{A}$ ($A$ holds at world $x$). 
 Sequents
take the form $\Gamma \vdash \Delta$ where $\Gamma$ and $\Delta$ are sets of
relational atoms of the form $\rel{x}{y}$ and labeled formulas of the form
$\lf{x}{A}$. 

In~\cite{DBLP:journals/jphil/Negri05}, it is shown how to automatically obtain
for some classes of axioms, including geometric implications, label systems for
extensions of the modal logic K with such axioms. The resulting calculus
exhibits good properties, including admissibility of the structural rules (when
the modal logic is an extension of classical logic),
context independence of rules, and admissibility of the cut rule.

\begin{observation}
    All the properties in~\Cref{def:rel} as well as 
    reflexivity, symmetry, totality and surjectivity are 
    all geometric implications~\cite{DBLP:journals/jphil/Negri05}.
\end{observation}

Following the procedure described
in~\cite{DBLP:journals/jphil/Negri05,DBLP:journals/apal/MarinMPV22} to
transform universal/geometric axioms into rules, we obtain  different
extensions of the system G3K for the modal logic K considering the modality
$\RBOX$ and the ``reverse'' modality $\LBOX$.

\begin{definition}[System $\syspast$]

	The rules for $\syspast$ are in~\Cref{fig:K}. In~\Cref{fig:ext}
we present the resulting relational rules for the properties of totality,
surjectivity, collusiveness, reflexivity and symmetry (for an arbitrary relation $R$),  and non-overlapping 
 and collectively connectedness  wrt two relations $R$ and $R'$. 

\end{definition}


\begin{figure}
\centering
\resizebox{.84\textwidth}{!}{
	$\begin{array}{cccc}
        \infer[I]{\seqGD{,\lf{x}{p}}{,\lf{x}{p}}}{}
        \qquad
        \infer[I]{\seqGD{,\rel{x}{y}}{,\rel{x}{y}}}{}
		\quad
            \infer[\wedge_L]{\seqGD{,x:A\wedge B}{}}{\seqGD{,\lf{x}{A}, \lf{x}{B}}{}}
             \quad 
            \infer[\wedge_R]{\seqGD{}{,\lf{x}{A\wedge B}}}{
            \deduce{\seqGD{}{,\lf{x}{A}}}{} &
            \deduce{\seqGD{}{,\lf{x}{B}}}{} 
            }
            \\\\
            \infer[\neg_L]{\seqGD{,\lf{x}{\neg A}}{}}{
                \seqGD{}{,\lf{x}{A}}
            }
            \quad
            \infer[\neg_R]{\seqGD{}{,\lf{x}{\neg A}}}{
                \seqGD{,\lf{x}{A}}{}
            }
            \quad
            \infer[\bot_L]{\seqGD{,\lf{x}{\bot}}{}}{}
\\\\
            \infer[\RBOX_L]{\seqGD{,\lf{x}{\RBOX A}, \rel{x}{y}}{}}{
                \seqGD{,\lf{y}{A},\lf{x}{\RBOX A}, \rel{x}{y}}{}
            }
            \qquad 
            \infer[\RBOX_R]{\seqGD{}{,\lf{x}{\RBOX A}}{}}{
                \seqGD{,\rel{x}{y}}{,\lf{y}{A}}
            }
            \quad
            \infer[\RDIAM_L]{\seqGD{,\lf{x}{\RDIAM A}}{}}{
                \seqGD{,\rel{x}{y}, \lf{y}{A}}{}
            }
            \qquad 
            \infer[\RDIAM_R]{\seqGD{,\rel{x}{y}}{,\lf{x}{\RDIAM A}}}{
                \seqGD{,\rel{x}{y}}{,\lf{x}{\RDIAM A},\lf{y}{ A}}
            }
			\\\\
            \infer[\LBOX_L]{\seqGD{,\lf{x}{\LBOX A}, \rel{y}{x}}{}}{
                \seqGD{,\lf{y}{A},\lf{x}{\LBOX A}, \rel{y}{x}}{}
            }
            \qquad 
            \infer[\LBOX_R]{\seqGD{}{,\lf{x}{\LBOX A}}{}}{
                \seqGD{,\rel{y}{x}}{,\lf{y}{A}}
            }
            \quad
            \infer[\LDIAM_L]{\seqGD{,\lf{x}{\LDIAM A}}{}}{
                \seqGD{,\rel{y}{x}, \lf{y}{A}}{}
            }
            \qquad 
            \infer[\LDIAM_R]{\seqGD{,\rel{y}{x}}{,\lf{x}{\LDIAM A}}}{
                \seqGD{,\rel{y}{x}}{,\lf{x}{\LDIAM A},\lf{y}{ A}}
            }
        \end{array}
    $}
    \caption{System G3K (rules for $\vee$ and $\imp$ omitted). In rules $\RBOX_R$, $\LBOX_R$, $\RDIAM_L$, and $\LDIAM_L$,  $y$ is fresh. \label{fig:K}}
\end{figure}

\begin{figure}
\centering
\resizebox{.80\textwidth}{!}{
	$\begin{array}{cccc}
            \infer[total]{\seqGD{}{}}{\seqGD{,\rel{x}{y}}{}}
            \quad
            \infer[\mathit{surj}]{\seqGD{}{}}{\seqGD{,\rel{y}{x}}{}}
            \quad
            \infer[\mathit{collusive}]{\seqGD{,\rel{x}{y},\rel{x'}{y},\rel{x}{z}}{}}{
                \seqGD{,\rel{x}{y},\rel{x'}{y},\rel{x}{z},\rel{x'}{z}}{}
            }
			\quad
            \infer[\mathit{refl}]{\seq{\Gamma}{\Delta}}{\seq{\Gamma, \rel{x}{x}}{\Delta}}
			\\\\
			\infer[symm]{\seq{\Gamma,\rel{x}{y}}{\Delta}}{
				\seq{\Gamma,\rel{y}{x}}{\Delta}}
		\quad
\infer[nover]{\seq{\Gamma, x Ry, x R'y }{\Delta}}{
}
	\quad
    \infer[cc]{\seq{\Gamma}{\Delta}}{
        \deduce{\Gamma, xR y \vdash \Delta}{}
        &
        \deduce{\Gamma, xR' y \vdash \Delta}{}
    }
\end{array}
$
}
\caption{Relational rules. In rules $total$
and $\mathit{surj}$, $y$ is fresh. Rule $\mathit{nover}$ (resp. $\mathit{cc}$)
corresponds to non-overlapping (resp. collectively correctness) for a relation $R$ 
w.r.t. a relation  $R'$ \label{fig:ext}.}
\end{figure}

In what follows, we write $\syspast_{S}$, where $S$ is a set of axioms, to
denote the extension of $\syspast$ with the corresponding rules. From
\cite{DBLP:journals/apal/MarinMPV22}, we know that extending $\syspast$ with
any subset of the considered axioms preserves cut-admissibility.

Now we illustrate the system $\syspast_{S}$ by giving an alternative proof of some of the
results presented in previous sections. 

\begin{example}
We can show that: 
    1) The axiom  scheme  $\axC$  is provable in $\syspast_{\{\mathit{collusive}\}}$; and 
	2) The system $\syspast_{\{\mathit{refl},\mathit{collusive}\}}$ can prove the modal axioms 
    $B$ (symmetry, $A \Rightarrow \RBOX \RDIAM A$), and $4$ (transitivity $\RBOX A \Rightarrow \RBOX \RBOX  A$)  i.e., 
	reflexive collusive relations are equivalence relations.
\end{example}
\begin{proof}
    We prove below each statement. Note that 
due to cut-admissibility, and invertibility of rules, the end sequents in each case 
 cannot be proved without using the relational rules in the 
extended system.\\

\begin{minipage}[t]{0.30\textwidth}
        \resizebox{\textwidth}{!}{
    $
        \infer[\Rightarrow_R]{\seq{}{\lf{x}{\RDIAM\LDIAM\RBOX A\Rightarrow\RBOX A}}}{
            \infer[\RBOX_R]{\seq{\lf{x}{\RDIAM\LDIAM\RBOX A}}{\lf{x}{\RBOX A}}}{
                \infer[\RDIAM_L]{\seq{\rel{x}{y}, \lf{x}{\RDIAM\LDIAM\RBOX A}}{\lf{y}{A}}}{
                    \infer[\LDIAM_L]{{\seq{\rel{x}{z}, \rel{x}{y}, \lf{z}{\LDIAM\RBOX A}}{\lf{y}{A}}}}{
                        \infer[\mathit{coll.}]{{\seq{\rel{x'}{z},\rel{x}{z}, \rel{x}{y}, \lf{x'}{\RBOX A}}{\lf{y}{A}}}}{
                            \infer[\RBOX_L]{{\seq{\rel{x'}{y}, \rel{x'}{z},\rel{x}{z}, \rel{x}{y}, \lf{x'}{\RBOX A}}{\lf{y}{A}}}}{
                                \infer[I]{{\seq{\cdots,\lf{x'}{\RBOX A},\lf{y}{A}}{\lf{y}{A}}}}{}
                            }
                        }
                    }
                }
            }
        }
    $}
    \end{minipage}
    \hfill
    \begin{minipage}[t]{0.3\textwidth}
        \centering
        \resizebox{\textwidth}{!}{
$
\infer[\Rightarrow_R]{\vdash x: A \Rightarrow \RBOX \RDIAM A}{
	\infer[\RBOX_R]{x:A \vdash x:\RBOX\RDIAM A}{
		\infer=[2\times\mathit{refl}]{\rel{x}{y}, x:A  \vdash y:\RDIAM A}{
			\infer[coll.]{\rel{x}{y}, \rel{x}{x}, \rel{y}{y}, x:A \vdash y:\RDIAM A}{
				\infer[\RDIAM_R]{\rel{y}{x}, \rel{x}{y}, \rel{x}{x}, \rel{y}{y}, x:A \vdash y:\RDIAM A}{
					\infer[I]{\cdots, x:A \vdash x: A}{
					}
				}
			}
		}
	}
}
$
}
    \end{minipage}
    \hfill
    \begin{minipage}[t]{0.3\textwidth}
        \centering
        \resizebox{\textwidth}{!}{
$
\infer[\Rightarrow_R]{\vdash x: \RBOX A \Rightarrow \RBOX \RBOX  A}{
	\infer=[2\times \RBOX_R]{x : \RBOX A \vdash x:  \RBOX \RBOX  A}{
		\infer[\mathit{refl}]{ \rel{x}{y}, \rel{y}{z}, x : \RBOX A\vdash z:    A}{
			\infer[coll.]{ \rel{y}{y},\rel{x}{y}, \rel{y}{z}, x : \RBOX A\vdash z: A}{
				\infer[\RBOX_L]{ \rel{x}{z}, \rel{x}{y}, \rel{y}{z}, \rel{y}{y}, x : \RBOX A\vdash z: A}{
					\infer[I]{ \cdots,x: \RBOX A, z : A\vdash z: A}{
					}
				}
			}
		}
	}
}
$
}
    \end{minipage}

\end{proof}

The authors in \cite{DBLP:journals/logcom/PedersenSA21}
give a modal characterization of the properties of balance and weak balance
in symmetric signed frames. 
For weak balance, the proposed axiom scheme 
is 
\[\axW: (\dplus \dplus A \imp \dplus A)\wedge ((\dplus\dminus A \vee \dminus\dplus A) \imp \dminus A )
\]
The semantics of the modal operators $\dplus$ and $\dminus$ 
is given in models of the form $\cM = \langle X, \rplus,\rminus, V\rangle$:
\[
    \begin{array}{lll}
        \cM ,x \models \dplus \phi & \mbox{ iff } & \mbox{ there exists } x' \mbox{ s.t. } x\rplus x'\mbox{ and  } \cM,x' \models \phi\\
        \cM ,x \models \dminus \phi & \mbox{ iff } & \mbox{ there exists } x' \mbox{ s.t. } x\rminus x' \mbox{ and  } \cM,x' \models \phi
    \end{array}
\]

Next we show  that  axiom $\axW$  can be proved  in the
extension of $\syspast$ 
with rules corresponding to the properties of
symmetric signed frames, and the rule enforcing that 
$\rplus$ is collusive (see \Cref{th:new-weak-balance}). 

\begin{example}
    Consider the system $\syspast_S$ where $S$ includes: 
    reflexivity of $\rplus$, symmetry of $\rplus$ and $\rminus$, 
    collusiveness of $\rplus$, and 
    non-overlapping and collectively connectedness of 
    $\rplus$ wrt $\rminus$. The system $\syspast_S$ 
     proves the axiom $\axW$. 
\end{example}
\begin{proof}
    We show below the derivation for the second conjunct in $\axW$:


    \begin{center}
        \resizebox{.95\textwidth}{!}{
    $
    \infer[\Rightarrow_R]{\vdash x: (\dplus\dminus A \vee \dminus\dplus A) \Rightarrow \dminus A}{
        \infer[\vee_L]{ x: \dplus\dminus A \vee \dminus\dplus A \vdash x:\dminus A}{ 
            \infer[\dplus_L]{x:\dplus\dminus A \vdash x:\dminus A}{
                \infer[\dminus_L]{x\rplus y, y:\dminus A \vdash x:\dminus A}{
                    \infer[cc]{y\rminus z, x\rplus y, z: A \vdash x:\dminus A}{
                        \infer[\dminus_R]{x\rminus z, \cdots z: A \vdash x:\dminus A}{
                            \infer[I]{\cdots z: A \vdash z: A}{}
                        }
                        &
                        \infer=[\mathit{symm}\times 2]{x\rplus z, y\rminus z, x\rplus y\cdots \vdash x:\dminus A}{
                            \infer[\mathit{refl}]{z\rplus x, y\rminus z, y\rplus x\cdots \vdash x:\dminus A}{
                                \infer[\mathit{coll}]{z\rplus z, z\rplus x, y\rminus z, y\rplus x, \cdots \vdash x:\dminus A}{
                                    \infer[\mathit{noover}]{ y\rplus z,  y\rminus z \cdots  \vdash x:\dminus A}{
                                    }
                                }
                            }
                        }
                    }
                }
            }
            &
            \infer[\dminus_L]{x:\dminus\dplus A \vdash x:\dminus A}{
                \infer[\dplus_L]{x\rminus y, y:\dplus A \vdash x:\dminus A}{
                    \infer[cc]{x\rminus y, y\rplus z,  z: A \vdash x:\dminus A}{
                        \infer[\dminus_R]{z\rminus x, \cdots,   z: A \vdash x:\dminus A}{
                            \infer[I]{ \cdots  z: A \vdash z: A}{}
                        }
            &
            \infer[\mathit{symm}]{z\rplus x, x\rminus y, y\rplus z\cdots \vdash x:\dminus A}{
                \infer[\mathit{refl}]{x\rplus z, x\rminus y, y\rplus z\cdots \vdash x:\dminus A}{
                    \infer[\mathit{coll}]{y\rplus y, x\rplus z, x\rminus y, y\rplus z\cdots \vdash x:\dminus A}{
                        \infer[\mathit{nover}]{ x\rplus y,  x\rminus y \cdots \vdash x:\dminus A}{}
                    }
                }
            }
                    }
                }
            }
        }
    }
    $
}
    \end{center}
\end{proof}

We can also  prove that the axiom characterizing balance
in~\cite{DBLP:journals/logcom/PedersenSA21} is provable in the system with
rules enforcing that both  $\rplus$ and $\rminus$ are collusive (see
\Cref{th:new-balance-strong}). Since the axioms
in~\cite{DBLP:journals/logcom/PedersenSA21} are also geometric, we can go in the other direction 
and show that the resulting system extended with the rules for those axioms can prove  the axioms of
collusiveness for $\rplus$ (and for $\rplus$ and $\rminus$ in the case of
strong balance).

%% file: conclusion.tex
In~\cite{joinet:hal-02369662,Joinet2021}, the agonistic (i.e., non-reflexive)
reading of non-equivalence collusions was explicitly introduced, giving rise to
the term \emph{collusion}. These works primarily aimed to show how collusion
theory extends the logical theory of abstraction~--based on equivalence
relations--~developed in the late 19th century (notably the \emph{definitions by
abstraction} of the Peano School, also known as Frege–Russell’s
\emph{abstraction principles}).
In this paper, we have established a strong connection between collusive relations and the
property of balance in networks. In addition to characterizing this property in
terms of collusive relations, we generalize the balance theorem by considering 
the non-symmetric case.
We conjecture that irreflexive collusive relations that are furthermore
confluent and co-confluent do characterize the (strong) balance property.
This would provide an additional generalization of our results, since in the
symmetric case, collusiveness, confluence, and co-confluence collapse, thereby
extending~\Cref{th:str-balance}.

Finally, we plan to investigate extensions of the dialogical semantic games
associated with the modal logic introduced here, in the spirit
of~\cite{DBLP:conf/lpar/FreimanOPF24}. This is motivated by the fact that the games 
 in ~\cite{DBLP:conf/lpar/FreimanOPF24} assume that 
the negative and positive relations are symmetric.